\documentclass[a4article, 12pt]{amsart}
\usepackage{palatino}
\usepackage{amsmath, amsfonts, amscd, amssymb}
\usepackage{graphicx}
\usepackage{tikz}
\usepackage{hyperref}

\usepackage[left=3.00cm, right=3.00cm]{geometry}

\newenvironment{myproof}[1][\unskip]{\paragraph{\textit{Proof #1.}} }{\qed}
\newenvironment{freeTextEnvironment}[1][\unskip]{ \par\noindent {\bf #1} \noindent }

\newtheorem{example}{Example}
\newtheorem{theorem}{Theorem}
\newtheorem*{theorem*}{Theorem}
\newtheorem{lemma}{Lemma}[section]
\newtheorem{proposition}{Proposition}[section]
\newtheorem{definition}{Definition}[section]

\newtheorem{corollary}{Corollary}[section]
\newtheorem{remark}{Remark}[section]

\newcommand{\Aut}{\mathrm{Aut}}

\newcommand{\bt}{\textbf{t}}

\newcommand{\bx}{\textbf{x}}
\newcommand{\bs}{\textbf{s}}
\newcommand{\SLC}{\mathrm{SL} (2,\mathbb C)}

\newcommand{\CC}{{\mathbb{C}}}

\newcommand{\HH}{{\mathbb{H}}}
\newcommand{\PP}{{\mathbb{P}}}

\newcommand{\ZZ}{{\mathbb{Z}}}

\newcommand{\p}{\partial }
\newcommand{\F}{{\mathcal F}}
\newcommand{\T}{{\mathcal T}}
\newcommand{\X}{{\mathcal X}}
\newcommand{\Z}{{\mathcal Z}}
\newcommand{\calO}{{\mathcal O}}

\newcommand{\M}[2]{ { \overline{\mathcal M}_{#1, #2} } }

\allowdisplaybreaks

\begin{document}

\title{$\rm SL(2, \mathbb C)$ group action on Cohomological field theories}
\date{\today}
\author{Alexey Basalaev}
\address{National Research University Higher School of Economics, Vavilova 7, 117312 Moscow, Russia}
\email{aabasalaev@edu.hse.ru}
\address{Ruprecht-Karls-Universit\"at Heidelberg, Germany}
\email{abasalaev@mathi.uni-heidelberg.de}

\maketitle

\begin{abstract}
We introduce the $\SLC$ group action on a partition function of a Cohomological field theory via a certain Givental's action. Restricted to the small phase space we describe the action via the explicit formulae on a CohFT genus $g$ potential. 
We prove that applied to the total ancestor potential of a simple elliptic singularity the action introduced coincides with the transformation of Milanov--Ruan changing the primitive form (cf. \cite{MR}). 

\end{abstract}
 
\section{Introduction}

  Cohomological field theories (CohFT for brevity) were introduced in the early 90s in \cite{KM}. They appeared to play an important role in many different subjects of mathematics --- they are key objects in the mirror symmetry conjectures, integrable hierarchies by \cite{DZ,FSZ} and geometry of the moduli space of curves \cite{PPZ}. An important tool to work with CohFTs that is used in all the aspects listed is Givental's action. However in some cases (and in singularity theory in particular) it does not give any feeling of the initial object geometry, where another --- $\SLC$ action is defined naturally. 
  
  \subsection{Cohomological Field Theories}
   Denote by $\M{g}{k}$ the moduli space of stable genus $g$ curves with $k$ marked points. Let $V$ be a $n$--dimensional $\CC$--vector space with a non--degenerate scalar product $\eta$. 
  A Cohomological field theory on the state space $(V, \eta)$ is a system of linear maps 
  $
    \Lambda_{g,k}: V^{\otimes k} \rightarrow H^*(\M{g}{k}, \mathbb C)
  $ 
  for all $g,k$ such that $\M{g}{k}$ exists and is non--empty. It is required to satisfy certain axioms that arise naturally from the geometry of the moduli space of curves. 
  
  By \cite{G,FSZ,Sh} there are two group actions on the Fock space of all partition functions of the CohFTs. 
  These two \textit{different} actions are now known as the actions of the \textit{upper--triangular} and \textit{lower--triangular} Givental groups or $R$-- and $S$--actions respectively.
  
  The genus zero data of every CohFT with unit defines a Frobenius manifold. The notion of Frobenius manifold was introduced by B. Dubrovin in 90s. It provides a generalization of the flat structures of K. Saito introduced in the early 80s and is crucial step for a total ancestor potential of a singularity.
  
  The structure of a Frobenius manifold $M$ is defined by the so--called Frobenius potential $F_M \in \mathbb C[[t^1, \dots, t^n]]$ that is subject to the system of non--linear  PDEs called WDVV equation.
  It was observed already by \cite{D} (cf. Exercise~B.3) that there is an affine group acting on the space of WDVV solutions, generated by translations 
  and inversion. 
 In this paper we investigate this affine group action from the point of view of Givental's action and its applications to singularity theory.
  
  \subsection{Modularity of CohFTs and Frobenius manifolds}
  The following examples should be considered as a motivation for our work.
  \begin{example}[Appendix~C in \cite{D}]\label{example: chazy}
    Consider the 3--dimensional Frobenius manifold structures on $\CC^2\times\HH$ with the pairing $\eta_{i,j} = \delta_{i+j,4}$, satisfying the quasi--homogeneity condition $E\cdot F_M = 2 F_M$ for $E = t^1 \frac{\p}{\p t^1} + \frac{1}{2} t^2 \frac{\p}{\p t^2}$. The potential of such a Frobenius manifold reads:
    \[
      F_M(\textbf t) = \frac{1}{2} (t^1)^2t^3 + \frac{1}{2} t^1 (t^2)^2 - \frac{(t^2)^4}{16} \gamma(t^3).
    \]
    The WDVV equation on $F_M$ is equivalent to the Chazy equation on $\gamma$. It reads: 
    $
      \gamma^{\prime \prime \prime} = 6 \gamma \gamma^{\prime\prime} - 9 (\gamma^\prime)^2.
    $
    Every solution $\gamma$ defines a Frobenius manifold potential.
    On the space of Chazy equation solutions there is a transitive $\SLC$ group action. 
    Hence there is a $\SLC$ action on the space of 3--dimensional Frobenius manifolds satisfying the quasi--homogeneity conditions as above.
    
    One can also apply a Givental's group element $R$ to a Frobenius manifold $M$ with the Frobenius potential $F_M(\bt)$ giving in general some new Frobenius manifold that will be denoted by $\hat R \cdot M$.
    However it's not clear at all how the analytical action of this example should be connected to the action of Givental.
  \end{example}
  
  Another important example comes from the Gromov--Witten theory.
  
  \begin{example}\label{example: ellipic GW}  Let $\X_N$ be the so--called elliptic orbifolds. $\X_2 := \PP^1_{2,2,2,2}$, $\X_3 := \PP^1_{3,3,3}$, $\X_4 := \PP^1_{4,4,2}$  and $\X_6 := \PP^1_{6,3,2}$. The explicit genus zero potentials $F_0^{\X_N}$ of the first two orbifolds were found in \cite{ST} and of the second two by the author (available at \cite{B_hp}). All these potentials $F_0^{\X_N}$ can be written 
  via the quasi--modular forms w.r.t the group $\Gamma(N) := \big\lbrace A \in {\rm SL}(2,\ZZ) \ | \ A \equiv \begin{pmatrix} 1 & 0 \\ 0 & 1 \end{pmatrix} {\rm mod}(N) \big\rbrace$ (see \cite{MR,MS2,SZ}).
  On a weight $k$ quasi--modular form $f(\tau)$ one can act with $A = \begin{pmatrix} a & b \\ c & d \end{pmatrix} \in \Gamma(N)$ by
  $f(\tau) \to \dfrac{1}{(c\tau + d)^k} f \left(\dfrac{a\tau+b}{c\tau+d}\right)$.
  However after applying this action to all quasi--modular forms, building up $F_0^{\X_N}$ 
  we will get a function that is not associated with any CohFT anymore (end even will not even be a solution to WDVV equation).
  One of the purposes of this paper is to define the action of $A$ on a genus $g$ potential of a CohFT, s.t. the quasi--modular forms involved are transformed as above.
  \end{example}

  \begin{example} By \cite{S} there is a Frobenius manifold structure on the base space $\mathcal{S}$ of a hypersurface singularity unfolding, that depends heavily on the additional choice of the so--called \textit{primitive form} $\zeta$ of Saito. 
  It was later found that at a point $\bs \in \mathcal{S}$ one can associate the CohFT partition function $\mathcal{A}_{\zeta,\bs}$ to a hypersurface singularity with the primitive form $\zeta$ fixed (see \cite{M} for the precise definition).
  
  In case of simple--elliptic singularities using certain Givental's action Milanov and Ruan gave in \cite{MR} the formula connecting $\mathcal{A}_{\zeta_1,\bs}$ and $\mathcal{A}_{\zeta_2,\bs^\prime}$ with a two different primitive forms $\zeta_1$,$\zeta_2$ of the same singularity and (in general) different points $\bs, \bs' \in \mathcal{S}$.
  
  In \cite{BT} the authors proposed particular $\SLC$ action on the space of $3$--dimensional Frobenius manifold as in Example~1 above
  to write down explicitly the effect of the primitive form change on the Frobenius manifold potential. However this action is written in a completely different form comparing to the formula of Milanov--Ruan and was not extended to the total ancestor potential. 
  
  In this paper we show that the two approaches agree by developing $\SLC$--action and its Givental's analog in the general setting (see Theorem~\ref{theorem: MR-BT equivalence}).
  \end{example}
  \subsection{Goals and results}
  Fix a CohFT $\Lambda_{g,k}$ on $V = \langle e_1,\dots, e_n\rangle$ with the pairing $\eta_{p,q} := \delta_{p+q,n+1}$ and unit vector $e_1$. Let $\Z = \exp\left(\sum_{g \ge 0} \hbar^{g-1} \F_g \right)$ be the partition function of it. The functions $\F_g$, called \textit{genus $g$ potentials}, are function of the formal complex variables $t^{d,\alpha}$ for all $1 \le \alpha \le n$ and $d \in \ZZ_{\ge 0}$. In particular, the variable $t^{0,\alpha}$ is associated with the insertions of vector $e_\alpha$ without a psi--class (see Section~\ref{section: GiventalsAction} for details.)
  
  Next to $e_1$, the pairing fixed gives the other distinguished vector --- $e_n$\footnote{We will comment on the basis and pairing fixing later in the text.}. 
  
  Let $A \in \rm SL(2, \mathbb C)$ act on $\CC$ by a linear--fractional transformation. Denote $A \cdot t^{0,n} := \dfrac{a t^{0,n} + b}{c t^{0,n} + d}$ for $A = \begin{pmatrix}a & b \\ c & d \end{pmatrix}$.
  In this paper we ``quantize'' this linear--fractional transformation 
  to the action $\hat A$ on the Fock space of all CohFT partition functions in the following way. 
  \\
  \begin{freeTextEnvironment}[The quantization conditions of $\SLC$--action:]
    \begin{itemize}\label{quantization conditions}
    \item The function $\tilde \Z := \hat A \cdot \Z$ is a CohFT partition function again,
    \item By the action $\hat A$ the variable $t^{0,n}$ of $\Z$ is transformed to $\tilde t^{0,n} := A \cdot t^{0,n}$ and $\tilde \Z$ is a function of $\tilde t^{0,n}$.
    \end{itemize}  
  \end{freeTextEnvironment}
  \noindent We only provide the certain quantization satisfying the conditions above and leave the question of uniqueness of such a quantization for the future work.
  
  Being a partition function, $\tilde\Z$ should also have an $\hbar$ expansion as $\Z$ does. For $\F_0$ and $\tilde\F_0$ being the coefficients of $\hbar^{-1}$ in $\Z$ and $\tilde \Z$ above, one concludes that the quantization condition above should hold \textit{in genus zero}. Namely, these functions define certain Frobenius manifolds $M$, $M^A$, and $\hat A$ should also act on the space of Frobenius manifolds and WDVV solutions. The quantization $\hat A$, we develop in this paper, is marked by the following result.
  
  \begin{theorem}[cf. \cite{B_o}]
    For $N=3,4,6$ and any $A \in \Gamma(N)$ the genus zero small phase space Gromov--Witten potential of $\X_N$ satisfies:
    $$
      \left( F_0^{\X_N} \right)^A = F_0^{\X_N}.
    $$
  \end{theorem}
  Namely, the genus $0$ small phase space potentials of $\X_N$ have explicit expressions via the quasi--modular forms w.r.t $\Gamma(N)$ (recall Example~\ref{example: ellipic GW}), but these potentials by their own are \textit{modular} w.r.t. $\Gamma(N)$ by the action $\hat A$ we introduce.
  In this text we show this theorem on the example of the orbifold $\X_2$ (not covered by \cite{B_o}).
\\
\\
  \noindent\textbf{Main results} of this paper are theorems~\ref{th:FrobeniusManifold} and~\ref{theorem: MR-BT equivalence}.
  
  In part~(a) of Theorem~\ref{th:FrobeniusManifold} we show the isomorphism between the Frobenius manifolds $M^A$ and $\hat R^\sigma \cdot M$ for the particular Givental's group element $R^\sigma$, depending on $A$.
  By this we get the bridge between the Dubrovin affine group action on the space of WDVV solutions
  and the particular Givental's action on a CohFT. It's important to note that the first one is easy to write down in the closed formula, but it's only applicable in genus $0$. The latter action is defined in all genera but hard to pack in a closed formula.

  In part~(b) of Theorem~\ref{th:FrobeniusManifold} we extend the $\SLC$ action to the higher genera providing the particular formulae for the $\SLC$--action on the genus $g$ small phase space potential of a CohFT. Its immediate corollary is that the Givental's action of $R^\sigma$ mentioned above satisfies the quantization conditions of the $\SLC$--action in full genera. The explicit formulae of Theorem~\ref{th:FrobeniusManifold} show in particular that we get a group action.
  
  In terms of partition functions of the CohFT's the $\SLC$--action we introduce appears to be equivalent to the \textit{composition} of one $R$--action and two $S$--actions. This is a surprising result by itself because the actions of the upper--triangular and lower--triangular groups do not commute. We discuss this in Section~\ref{section: coordinateFree}.
  
  Theorem~\ref{theorem: MR-BT equivalence} can be considered as a main theorem of this work. It shows that the $\SLC$--action introduced is equivalent to the primitive form change for simple--elliptic singularities.

  \subsection{Applications in mirror symmetry}
  In the context of \textit{global mirror symmetry} (cf. \cite{CIR}) a B--model is treated globally, so that one can take it ``at the different phases''. Two different phases of the same B--model can give two different mirror A--models. 
  It's conjectured then that the partition functions of these two A models are connected by the certain Givental's action. This conjecture is known under the name of \textit{CY/LG correspondence conjecture}.
  
  It was proved in \cite{BP} and \cite{B6dim} that the $\SLC$--action we introduce in this paper gives indeed the desired CY/LG correspondence for the simple--elliptic singularities. Namely, the B--model is given by $\mathcal{A}_{\zeta,\bs}$ of a simple--elliptic singularity (with the primitive form playing the role of the ``phase'') and the A--models are given by the Gromov--Witten theory of $\X_N$ and the so--called FJRW theory of the same singularity with a group action.
  It turns out that the analytic $\SLC$--action is easy to compute in these cases, and by Theorem~\ref{th:FrobeniusManifold} we get the particular Givental's action corresponding to it. This result was expected due to Theorem~\ref{theorem: MR-BT equivalence} and some ideas of the global mirror symmetry. 
  
  \subsection{Organization of the paper}
  In Section~\ref{section: AnalyticalQuantization} we give an analytical approach to the $\SLC$ action on the genus $0$ part of the CohFT and show the modularity of the GW--theory of $\PP^1_{2,2,2,2}$ under this action. 
  We recall basic facts about Givental's action in Section~\ref{section: GiventalsAction}. 
  In Section~\ref{section: aaa} we define Givental's action analog of the analytical $\SLC$ action.
  In Section~\ref{section: coordinateFree} we write the $\SLC$ action in terms of Givental's group action only.
  We discuss singularity theory applications in Section~\ref{section: totalAncestorPotential}.
  
  \subsection{Acknowledgement}
  The author is grateful to Sergey Shadrin for his help with Givental's action, Claus Hertling for many useful comments and to Davide Veniani for the editorial help. The author is also very grateful to Maxim Kazarian for sharing his unpublished notes and to the anonymous referees for many valuable comments.

  \section{Analytical quantization of $\SLC$ action}\label{section: AnalyticalQuantization}
    In this section we develop an analytical approach to the quantization of the $\SLC$ action on the genus zero data of the CohFTs  --- Frobenius manifolds.
    \subsection{Frobenius manifolds}\label{section: FrobeniusManifolds}
    Let $M$ be a connected complex manifold. Assume its tangent sheaf $\T_M$ to be endowed with a non--degenerate $\calO_M$--symmetric bilinear form $\eta$ and an associative and commutative $\calO_M$--algebra structure. Let $\circ$ stand for the product of this algebra and $e$ for its unit vector field. In what follows we consider the metric on $M$ defined by $\eta$. Let $\nabla$ stand for its Levi--Civita connection.
    
    \begin{definition}
      The data $(M,\circ,\eta, e)$ is called a \textit{Frobenius manifold} if the following conditions are satisfied. 
      The metric $\eta$ is flat, $\nabla e = 0$, $\eta$ and $\circ$ satisfy the Frobenius algebra property: $\eta(u \circ v, w) = \eta(u,v \circ w)$ for any $u,v,w \in \T_M$, the tensor field $\nabla_z \left( \eta(u \circ v, w) \right)$ for any $u,v,w,z \in \T_M$ is symmetric in all four components.      
    \end{definition}
    \begin{remark}
      Usually one requires also the product $\circ$ to satisfy some quasi--homogeneity condition, introducing a special $\T_M$ element, called \textit{Euler vector field}. We do not raise such a condition in this paper, assuming a larger class of Frobenius manifolds.
    \end{remark}

    Note that different points of a Frobenius manifold $M$ have generally different algebra structures on their tangent spaces. In many cases working with a Frobenius manifolds one uses only a local data of it. Moreover in some situations only the germ of a Frobenius manifold is defined naturally. However in the ``good'' special cases (like for example in mirror symmetry) a Frobenius manifold appears to be defined \textit{globally}, so that it's reasonable to consider the algebra structures at the different points of $M$. 
    
    In order to work with a Frobenius manifold locally consider the following data.
    It follows from the definition of a Frobenius manifold that locally at every point $p \in M$ there is a system of \textit{flat coordinates} $t^1, \dots, t^n$, in which the metric $\eta$ has constant components. Associate a basis of $T_p M$ with the vectors $\p/\p t^i$ and consider $\eta_{ij}$ as components of $\eta$ in this basis. We also assume that $\p/\p t^1$ is the unit of $\circ_p: T_pM \otimes T_pM \to T_pM$.
    
    It follows from the definitions of a Frobenius manifold that there is a function $F(\bt) = F(t^1, \dots, t^n)$, represented by a convergent power series in $t^1, \dots, t^n$, s.t. the structure constants $c_{ij}^k$ of the product $\circ_p$ satisfy:
    $$
      c_{ij}^k(\bt) := \sum_{p=1}^n \frac{\p^3 F}{\p t^i \p t^j \p t^p} \eta^{pk}, \quad 1 \le i,j,k \le n, 
    $$
    where $\eta^{ij} := \sum_{p,q} \eta_{pq} \delta^{pi} \delta^{qj}$. Because of the Frobenius algebra property and choice of the coordinates we have $\eta_{ij} = c_{1ij} := \partial_{t^1}\partial_{t^i}\partial_{t^j} F$.
    From the associativity and commutativity of the product one deduces that the function $F(\bt)$ satisfies WDVV equation --- for every fixed $1 \le i,j,k,l \le n$ holds:
    $$
      \sum_{p,q} \frac{\partial^3 F }{\partial t^i \partial t^j \partial t^p} \ \eta^{pq} \ \frac{\partial^3 F}{\partial t^q \partial t^k \partial t^l} 
      = 
      \sum_{p,q} \frac{\partial^3 F}{\partial t^i \partial t^k \partial t^p} \ \eta^{pq} \ \frac{\partial^3 F}{\partial t^q \partial t^j \partial t^l},
    $$
    The function $F(\bt)$ is called \textit{potential\footnote{in some articles this function could be also called ``prepotential''.} of the Frobenius manifold $M$}.

    Sometimes we are given first a function $F$ satisfying WDVV equation without any underlying manifold $M$ and convergence property. In these occasions $F$ could anyway define a (germ of) a Frobenius manifold that is called \textit{formal}. We will drop this word assuming it to be clear from the context.
    
    \begin{definition}
      Two Frobenius manifolds $M_1$ and $M_2$ are called (locally) isomorphic if there is a diffeomorphism $\phi: M_1 \to M_2$ such that for some fixed $\bt \in M_1$ and $\phi(\bt) \in M_2$ holds:
      \begin{itemize}
       \item $\phi$ is linear conformal transformation of the metrics of $M_1$ and $M_2$,
       \item the differential of $\phi$ is an isomorphism of the algebras $T_{\bt}M_1$ and $T_{\phi(\bt)}M_2$.
      \end{itemize}
      In this case we write $M_1 \mid_\bt \ \cong \ M_2 \mid_{\phi(\bt)}$.
    \end{definition}
   
    We define now the $\SLC$--action on the space of Frobenius manifolds.
    Let the Frobenius manifold potential $F(\bt)$ have the form
    \begin{equation}\label{equation: Frob potential general form}
      F(\bt) = F(t^1, \dots, t^n) := \frac{1}{2} (t^1)^2 t^n + \frac{1}{2} t^1 \sum_{k=2}^{n-1} t^p t^{n+1-k} + H(t^2, \dots, t^n),
    \end{equation}
    for some function $H(t^2, \dots, t^n)$ not depending on $t^1$.
    Potential $F(\bt)$ defines a formal Frobenius manifold with the metric $\eta_{i,j} = \delta_{i+j,n+1}$ (note however that our results are easily translated to a more general choice of $\eta$ except when $\eta_{1,1} \neq 0$).
    \begin{definition}
      For any $A \in \SLC$ define the function:
      \begin{equation}\label{eq:SLAction}
	\begin{aligned}
	  F^A(\bt)  := \frac{1}{2} (t^1)^2t^n &+ \frac{1}{2}t^1 \sum_{k=2}^{n-1} t^k t^{n+1-k} + \frac{c}{8(ct^n + d)} \left(t^2 t^{n-1} + \dots + t^{n-1}t^2 \right)^2
	  \\
	  &+ (ct^n+d)^2 H \left(\frac{t^2}{ct^n+d}, \dots, \frac{t^{n-1}}{ct^n+d}, \frac{at^n + b}{ct^n+d} \right).
	\end{aligned}
      \end{equation}
    \end{definition}
    
    \begin{remark}
      In the definition above we have applied the change of the variables $\bt \to \tilde \bt$ to the function $F^A$ above, s.t. $\tilde t^n = A \cdot t^n$. Such a special choice of the variable $t^n$ is made because 
      the corresponding vector $\p/\p t^n$ satisfies $\eta(\p/\p t^1, \p/\p t^n) = 1$. However we could have picked any other vector $v = \sum_{k=2}^n a_k \p/\p t^k$, not lying in the kernel of $\eta(\p/\p t^1, \bullet)$ and the corresponding variable. Then applying a linear change of the variables we would land in the setting we used above.
    \end{remark}
    
    It was observed by Dubrovin (cf. Appendix~B in \cite{D}) that there is a non--trivial symmetry of WDVV equation, called ``Inversion transformation'' such that for $F(\bt)$ as above the function $F^I(\hat \bt)$ solves WDVV too:
    \begin{equation*}\label{equation: Inversion}
      \hat I \cdot F := F^I (\hat \bt ) = (t^n)^{-2} \left[ F(\bt) - \frac{1}{2} t^1 \sum_{k=0}^n t^k t^{n+1-k} \right],
    \end{equation*}
    where 
    $\hat t^1 := \sum_k {t^k t^{n+1-k}}/(2t^n)$, $\hat t^n := -1/t^n$ and $\hat t^\alpha := t^\alpha/t^n$ for all $1 < \alpha < n$.
\\
\\
    In the definition above it's clear that if $c=0$, the function $F^A$ differs from the function $F$ just by a linear change of the variables, defining therefore an isomorphic Frobenius manifold structure. In what follows we assume silently only such $A \in \SLC$, that $c \neq 0$.
    
  \subsection{Analytical quantization via composition}
  Fix a matrix $A \in \SLC$. Let $z$ be a coordinate on $\CC$. For an arbitrary function of $z$, the change of the variables $z \to A \cdot z := \frac{a z + b}{c z + d}$ can be written via the following composition.
  \[
    A \cdot z = \frac{a z + b}{c z + d} = T_1 \cdot S_{c^{2}} \cdot I \cdot T_2 \cdot z,
  \]
  where $T_1(z) = z + a/c$, $S_{c^{2}}(z) = c^{2} z $, $T_2(z) = z + d/c$ and $I(z) = -1 / z$. Our aim now is to quantize these changes of the variables to the operators on the space of WDVV equation solutions.
  
  Let $L(z)$ be a linear change of the variables. For a WDVV solution $F = F(t^1,\dots,t^n)$ let $\hat L \cdot F := F(t^1,\dots,t^{n-1},L(t^n))$. It's clear that $\hat L \cdot F$ is a solution to WDVV equation too. Because $T_1,T_2$ and $S_{c^2}$ are linear changes of the variables, we can immediately quantize them to the operators $\hat T_1, \hat T_2$ and $\hat S_{c^2}$ acting on the space of WDVV equation solutions.
  Finally, the operator $I$ is quantized by the Inversion transformation $\hat I$ of Eq.\eqref{equation: Inversion}. Assuming $A$, $T_1,T_2$ and $S_{c^2}$ to act on the domain of definition of the flat coordinate $t^n$ we introduce the definition.
  
  \begin{definition}Define the action $\hat A$ on the space of WDVV solutions:
    $$
      \hat A := \hat T_1 \cdot \hat S_{c} \cdot \hat I \cdot \hat S_{c^{-1}} \cdot \hat T_2.
    $$
  \end{definition}
  \begin{proposition}\label{prop:Composition} For the $\hat A$ as above we have:
    \begin{itemize}
      \item The action $\hat A$ satisfies the quantization condition in genus $0$.
      \item The action $\hat A$ agrees with the formula \eqref{eq:SLAction} up to quadratic terms:
	$$
	  \hat A \cdot F = F^A + \text{quadratic terms}.
	$$
    \end{itemize}
    In particular the function $F^A$ is a solution to WDVV equation.
  \end{proposition}
  \begin{myproof}
    The quantization presented coincides with the action of $A$ on $t^n$ by the construction and also $\hat A$ acts on the space of WDVV solution as the composition of operators acting of the space of WDVV solutions. 

    For the second part note that applying $\hat S_{c} \cdot \hat I \cdot \hat S_{c^{-1}}$ to $H(t^2, \dots, t^n)$ both $S_c$ and $S_{c^{-1}}$ add their factors to the cubic terms defining the pairing, but they cancel out. 
    The rest is an easy check.
  \end{myproof}

  This approach to the affine group action on the space of WDVV equation solutions was proposed already by \cite{D} (cf. Exercise~B.3) via certain special change of the variables. 

  The following section shows that the $\SLC$--action $F \to F^A$ is important and natural in the certain cases.
  
  \subsection{Example: Gromov--Witten theory of $\PP^1_{2,2,2,2}$}\label{subsection: GW}
  Consider the so--called theta--constants $\vartheta_2(\tau),\vartheta_3(\tau),\vartheta_4(\tau)$ that are the values at $z = 0$ of the Jacobi theta functions $\vartheta_k(z,\tau)$. 
  Let $X_k^\infty(\tau)$ 
  be the logarithmic derivatives of the theta--constants: 
  $$X_k^\infty(\tau) := 2 \frac{\partial}{\partial \tau} \log \vartheta_k(\tau)$$
  Well--known fact is that the functions $(\vartheta_2(\tau))^2$, $(\vartheta_3(\tau))^2$, $(\vartheta_4(\tau))^2$ are modular forms of weight $1$ w.r.t. $\Gamma(2) = \{A \in \mathrm{SL}(2,\ZZ) \ | \ A \equiv \mathrm{Id} \ \text{ mod } 2\}$. In particular holds:
  \begin{equation*}
    \frac{1}{(c\tau+d)} \left(\vartheta_k\left(\frac{a\tau+b}{c\tau+d}\right) \right)^2 = \left( \vartheta_k(\tau) \right)^2 \text{ for any } 
    \begin{pmatrix}
      a &b\\
      c & d
    \end{pmatrix}
    \in \Gamma(2), \ k=2,3,4.
  \end{equation*}
  It's clear from the definition that the functions $X_k^\infty(\tau)$ also satisfy certain corrected --- \textit{quasi}--modularity property.
  
  Frobenius manifold potential of the Gromov--Witten theory of the orbifold $\PP^1_{2,2,2,2}$ was found in \cite{ST} to be\footnote{Following the standard convention we use the lower case variables numbering $t_k$ in the explicit example rather than $t^k$ used to the general formulae.}
  \begin{align*}
    F&^{\PP^1_{2,2,2,2}}_0 (\bt) = \frac{1}{2} t_1^2 t_6 + \frac{1}{4} t_1 \left( \sum_{i=2}^5 t_i^2 \right) + \frac{1}{8} (t_2t_3t_4t_5) \left( X_3^\infty - X_4^\infty \right) \\
    & - \frac{1}{192} \left(\sum_{i=2}^5 t_i^4 \right) \left( 4X_2^\infty + X_3^\infty + X_4^\infty \right) - \frac{1}{32} \left( \sum_{2 \le i<j \le 5 } t_i^2t_j^2 \right) \left( X_3^\infty + X_4^\infty\right),
  \end{align*}
  where we use $X_k^\infty := X_k^\infty(t_6)$.
  It's straightforward to note that one can put this potential to the form of Eq.\eqref{equation: Frob potential general form} by applying the linear change of the variables $t_2,\dots,t_5$. This allows us to use the formula~\eqref{eq:SLAction}.
  
  \begin{proposition}
    Let $A \in \Gamma(2)$, consider the action of it by formula~\eqref{eq:SLAction}. Then we have:
    $$
      \left( F_0^{\PP^1_{2,2,2,2}} \right)^A = F_0^{\PP^1_{2,2,2,2}}.
    $$
    Namely, the additional summand and meromorphic factors of formula~\eqref{eq:SLAction} are all adsorbed by the quasi--modularity property of $X_k^\infty$.
  \end{proposition}
  \begin{myproof}
    This is obtained by easy computations and the modularity formula above.
  \end{myproof}

   It's important to note that taking $A \in \mathrm{SL}(2,\ZZ)$ such that $A \not\in \Gamma(2)$ the genus zero potential of $\PP^1_{2,2,2,2}$ is transformed differently. 
   Taking $A = \begin{pmatrix} 1 & 1 \\ 0 & 1 \end{pmatrix}$ or $A = \begin{pmatrix} 0 & -1 \\ 1 & 0 \end{pmatrix}$ we get $ ( F_0^{\PP^1_{2,2,2,2}}(\bt) )^A = F_0^{\PP^1_{2,2,2,2}}(\tilde \bt)$ for $\tilde \bt$ differing from $\bt$ by the permutations of the variables $t_2,\dots,t_5$. Hence the Frobenius manifold obtained by such an action is still isomorphic to the initial one. We can not expect such a behavior for an arbitrary $A \in \SLC$.
  
  \subsection{Example: Hurwitz--Frobenius manifolds}
  Consider the space of meromorphic functions $\lambda: C \to \mathbb P^1$ on the compact genus $g$ Riemann surface $C$. Fix the pole orders of $\lambda$ to be $\textbf k := \{k_1, \dots, k_m\}$:
    $$
      \lambda^{-1}(\infty) = \{ \infty_1, \dots, \infty_m \}, \quad \infty_p \in C,
    $$
    so that locally at $\infty_p$ we have $\lambda(z) = z^{k_p}$. 

    Such meromorphic function defines a ramified covering of $\mathbb P^1$ by $C$ with the ramification profile $\textbf k$ over $\infty$. Assume further that $\lambda$ has only simple ramification points at $P_i \in \mathbb P^1 \backslash \{0\}$. On the space of the pairs $(C, \lambda)$, considered up to a certain equivalence, B. Dubrovin introduced in \cite[Lecture 6]{D} a Frobenius manifold structure that is now known under the name \textit{Hurwitz--Frobenius} manifold and is denoted by $\mathcal H_{g; \textbf{k}}$.

    When $g=1$ the ramified covering $\lambda$ is written via the elliptic functions and one of the parameters of it (and hence of the Hurwitz--Frobenius manifold) is $\tau \in \HH$, that stands for the modulus of an elliptic curve. For $\textbf{k} = \{2,2,2,2\}$ it has the following form:
    $$
      \lambda(z) = \sum_{i = 1}^4 \left( \wp(z- a_i, \tau) u_i + \frac{1}{2} \frac{\wp^\prime(z - a_i, \tau)}{\wp(z-a_i, \tau)} s_i \right) + c,
    $$
    where $\wp(z,\tau)$ is the Weierstrass function and $a_i$,$u_i$,$s_i$,$c$ are complex parameters. The corresponding Frobenius potential is also written in terms of a certain quasi--modular forms (see \cite{B_hf}). 
    One can consider the Frobenius manifold structure on $\mathcal{H}_{1,\textbf{k}}$ at different points $p_1$ and $p_2$. 
    Because one of the parameters of $p_1$ and $p_2$ is the modulus of the corresponding elliptic curve, it's natural for two to be connected by a certain $\SLC$ action.
    
    Given a Frobenius potential $F_1(\bt)$ encoding the algebra structure at $p_1$ one can consider $F_2(\bt) := F_1(\bt + p_2-p_1)$. However such a shift applied to the function $f(z)$, holomorphic in $\HH$ (like for example the functions $X_k^\infty(\tau)$), reduces drastically the domain of the holomorphicity. In order to keep the domain of holomorphicity big, one should apply not the Taylor series shift, but the following action instead (\cite{Z}):
    $$
      f(z) \to f \left( \frac{\tau_0 - \bar \tau_0 z}{1-z} \right).
    $$
    This action can be realized by the composition of the rescaling and the $\SLC$ action developed above.

  \section{Cohomological field theories and Givental's action}\label{section: GiventalsAction}
  
    We briefly recall some basic facts about the CohFTs and introduce Givental's action in the infinitesimal form computed in \cite{L}. 
    
    \subsection{Cohomological Field Theory axioms}
      Let $(V,\eta)$ be a finite--dimensional vector space with a non--degenerate bilinear form on it.
      Consider a system of linear maps 
      $$
	\Lambda_{g,k}: V^{\otimes k} \rightarrow H^*(\M{g}{k}),
      $$ 
      defined for all $g,k$ such that $\M{g}{k}$ exists and is non--empty.
      It is called Cohomological field theory on $(V,\eta)$ if it satisfies the following axioms.
      
      \begin{itemize}
	\item[\textbf A1:] $\Lambda_{g,k}$ is equivariant w.r.t. the $S_k$--action permuting the factors in the tensor product and the numbering of marked points in $\M{g}{k}$.
	\item[\textbf A2:] For the gluing morphism $\rho:\M{g_1}{k_1+1}\times\M{g_2}{k_2+1} \rightarrow \M{g_1+g_2}{k_1+k_2}$ holds:
	$$
	  \rho^* \Lambda_{g_1+g_2,k_1+k_2} = (\Lambda_{g_1, k_1+1} \cdot \Lambda_{g_2,k_2+1}, \eta^{-1}),
	$$
	where we contract with $\eta^{-1}$ the factors of $V$ that correspond to the node in the preimage of $\rho$.
	\item[\textbf A3:] For the gluing morphism $\sigma: \M{g}{k+2} \rightarrow \M{g+1}{k}$ holds:
	$$
	  \sigma^* \Lambda_{g+1,k} = (\Lambda_{g, k+2}, \eta^{-1}),
	$$
	where we contract with $\eta^{-1}$ the factors of $V$ that correspond to the node in the preimage of $\sigma$.
      \end{itemize}
      In this paper we further assume that the CohFT $\Lambda_{g,k}$ is unital --- there is a fixed vector $\textbf 1 \in V$ called \textit{unit} such that the following axioms are satisfied.
    
      \begin{itemize}
	\item[\textbf U1:] For every $a,b \in V$ we have: $\eta(a, b) = \Lambda_{0,3}(\textbf{1} \otimes a \otimes b)$.
	\item[\textbf U2:] Let $\pi: \M{g}{k+1} \rightarrow \M{g}{k}$ be the map forgetting the last marking, then:
	  $$
	    \pi^* \Lambda_{g,k}(a_1 \otimes \dots \otimes a_k) = \Lambda_{g,k+1}(a_1 \otimes \dots \otimes a_k \otimes \textbf{1}).
	  $$
      \end{itemize}
      In what follows we will denote the CohFT just by $\Lambda$ rather than $\Lambda_{g,k}$ when there is no ambiguity.
      
      We associate to a CohFT certain generating function, called partition function of a CohFT. To do this we have to fix a basis on $V$. Let $e_1,\dots,e_n$ be the basis vectors, s.t. $e_1$ coincides with the unit $\textbf 1$ above.      
      Let $\psi_l \in H^2(\M{g}{k})$ for $1 \le l \le k$ be so--called \textit{psi--classes} (namely, $\psi_l$ is the first Chern class of the line bundle, whose fiber over a $[C] \in \M{g}{k}$ is the cotangent line to $C$ at the $l$--th marked point). The genus $g$ correlators of the CohFT are the following numbers:
      $$
	\langle \tau_{d_1}(e_{\alpha_1}) \dots \tau_{d_k}(e_{\alpha_k}) \rangle_g := \int_{\M{g}{k}} \Lambda_{g,k}(e_{\alpha_1} \otimes \dots \otimes e_{\alpha_k}) \psi_1^{d_1} \dots \psi_k^{d_k}.
      $$
      Whenever $2 -2g-l < 0$ the correlators of a unital CohFT satisfy Dilaton equation:
      \begin{equation*}
	\left\langle \tau_1(e_1) \prod_{k=1}^l \tau_{d_k}(e_{\alpha_k}) \right\rangle_g = (2g-2+l) \left\langle \prod_{k=1}^l \tau_{d_k}(e_{\alpha_k}) \right\rangle_g
      \end{equation*}
      and String equation:
      \begin{equation*}
	\left\langle \tau_0(e_1) \prod_{k=1}^l \tau_{d_k}(e_{\alpha_k}) \right\rangle_g = \sum_{m = 1, d_m \neq 0}^l \left\langle \prod_{k \neq m} \tau_{d_k}(e_{\alpha_k}) \cdot \tau_{d_m -1}(e_{i_m}) \right\rangle_g
      \end{equation*}
      Denote by ${\mathcal F}_g$ the generating function of the genus $g$ correlators:
      $$
	{\mathcal F}_g := \sum_{k=1}^\infty \sum_{\boldsymbol \alpha, \bf d} \frac{\langle \tau_{d_1}(e_{\alpha_1}) \dots \tau_{d_k}(e_{\alpha_k}) \rangle_g}{\Aut( \{ \boldsymbol \alpha, \bf d \})} \ t^{d_1,\alpha_1} \dots t^{d_k,\alpha_k},
      $$
      where for all $1 \le l \le k$ each $\alpha_l$ runs from $1$ to $n$ and each $d_l$ runs from $1$ to $\infty$ with the additional condition that the monomial $t^{d_1,\alpha_1} \dots t^{d_k,\alpha_k}$ appears in the sum only once.
      
      We will call $\F_g$ \textit{genus $g$ potential} of the CohFT. We also use the notation $F_g(\bt) = F_g(t^1,\dots,t^n) := \F_g \mid_{t^{d,\alpha} = 0 \ d \ge 1, \ t^\alpha := t^{0,\alpha}}$, called \textit{small phase space potential}. 
      It is useful to assemble the correlators into a generating function called partition function of the CohFT:
      $$
	\mathcal{Z} := \exp \left( \sum\nolimits_{g \ge 0} \hbar^{g-1} \mathcal F_g \right).
      $$
      \begin{remark}
	Because $\Lambda_{g,k}$ is a system of linear maps, it's clear that another choice of the basis on $V$ will give a genus $g$ potential, differing from the previous one only by a linear change of the variables. Due to this fact we will not refer to the different choice of the basis in what follows.
      \end{remark}
      It follows from the certain topological properties on $\M{0}{k}$ that the function $F_0(\bt)$ is solution of WDVV equation and defines a Frobenius manifold on the small neighborhood of the origin in $\CC^n$. However in the particular examples this Frobenius manifold structure can be extended to a larger domain when $F_0$ turns out to have big radius of convergence. This is best illustrated by the Gromov--Witten theory of $\X_2$, defining the Frobenius manifold structure on $\CC^5\times\HH$.

    \subsection{Infinitesimal version of Givental's action}      
      \textit{Upper--triangular group} consists of all elements $R = \exp(\sum_{l=1}^\infty r_lz^l)$, where:
	$$
	  r(z) = \sum_{l \ge 1} r_l z^l \in {\rm Hom}(V,V) \otimes \mathbb C[z], \quad r(z) + r(-z)^* = 0.
	$$ 
	Following Givental, for  $(r_l)^{\alpha,\beta} = (r_l)^\alpha_\sigma \eta^{\sigma, \beta}$ we define:
	\begin{equation*}
	  \begin{aligned}
	    \widehat {r_lz^l} := & -(r_l)_1^\alpha \frac{\partial}{\partial t^{l+1,\alpha}} + \sum_{d =0}^\infty t^{d, \beta} (r_l)_\beta^\alpha \frac{\partial}{\partial t^{d+l, \alpha}}
	    \\
	    & + \frac{\hbar}{2} \sum_{i+j=l-1} (-1)^{i+1} (r_l)^{\alpha,\beta} \frac{\partial^2}{\partial t^{i,\alpha} t^{j,\beta}}
	    \quad \text{for every} \quad l \ge 1.
	  \end{aligned}
	\end{equation*}
	\textit{Lower--triangular group} consists of all elements $S = \exp( \sum_{l=1}^\infty s_l z^{-l})$ where:
	$$
	  s(z) = \sum_{l \ge 1} s_l z^{-l} \in {\rm Hom}(V,V) \otimes \mathbb C[z^{-1}], \quad s(z) + s(-z)^* = 0.
	$$ 
	Following Givental, we define:
	\begin{align*}
	  \sum_{l=1}^\infty (s_lz^{-l})\widehat{\ } & :=  -(s_1)_1^\alpha \frac{\partial}{\partial t^{0,\alpha}}
	  + \frac{1}{\hbar} \sum_{d=0}^{\infty} (s_{d+2})_{1,\alpha} \, t^{d,\alpha}
	  \\ 
	  & + \sum_{ \substack{d=0\\ l=1} }^\infty
	  (s_l)_\beta^\alpha \, t^{d+l,\beta} \frac{\partial}{\partial t^{d,\alpha}}
	  + \frac{1}{2 \hbar} \sum_{ \substack{d_1,d_2 \\ \alpha_1,\alpha_2} }
	  (-1)^{d_1} (s_{d_1+d_2+1})_{\alpha_1,\alpha_2} \, t^{d_1,\alpha_1} t^{d_2,\alpha_2}.\notag
	\end{align*}
	Note that $S_0 \in \mathrm{Hom}(V,V)$ does not belong to a lower--triangular group except when $S_0 = \mathrm{Id}$. Later on we will make certain sense of the $S_0$ action too treating it exclusively.
	
	\begin{definition}
	  The action of the differential operators $\hat R := \exp( \sum_{l=1}^\infty \widehat{r_l z^l})$ and $\hat S := \exp( \sum_{l=1}^\infty (s_lz^{-l})\widehat{\ })$ on the partition function of the CohFT is called Givental's $R$--action and $S$--action respectively.
	\end{definition}
	Our interest in Givental's action comes from the following theorem.
	\begin{theorem}[cf. \cite{G,FSZ,Sh}]
	  Let $R = R(z)$ and $S = S(z)$ be some elements of the upper-- and lower--triangular groups of Givental respectively. Then the operators $\hat R$ and $\hat S$ act on the space of partition functions of the cohomological field theories.
	\end{theorem}

	\begin{remark}
	  The $R$-- and $S$--actions form two \textit{different} group actions, whose commutator is in general non--trivial.
	\end{remark}
	
    We can also consider the action of the Givental's group element $R$ on the Frobenius manifold $M$ with the potential $F(\bt)$. Doing this we act first on the CohFT partition function $\mathcal Z$ and consider the restriction to the small phase space of the (new) genus zero potential after applying the $R$--action.    
    \begin{definition}\label{definition: RonM}
      Let $M$ be a Frobenius manifold defined by a CohFT $\Lambda_{g,n}$ with the partition function $\mathcal Z$. Let $R$ be a Givental's group element. We denote by $\hat R \cdot M$ the Frobenius manifold given by the Frobenius potential 
      $$
	\hat R \cdot F := [\hbar^{-1}] \log \left( \hat R \cdot \mathcal{Z} \right) \mid_{t^{k,\alpha} = 0, \ k \ge 1}.
      $$
    \end{definition}

    \begin{remark}\label{remark: calibration}
      In this paper we work in the setting of cohomological field theories. Namely, the potential $\F_0$ has by definition no quadratic terms because $\M{0}{k}$ is only defined for $k \ge 3$ and $\F_1$ has no constant term because $\M{1}{k}$ is only defined for $k \ge 1$. This ``missing'' data could be assumed too, packed in the so--called \textit{calibration} of a CohFT. 
      At the same time, it's clear that we don't miss ``to much'' not taking the calibration into account. For example, the Frobenius manifold structure is not affected by these quadratic terms discussed.
      
      The calibration of a CohFT is affected by the $S$--action. It's clear from the formula above that the $S$--action can introduce the quadratic terms to $\F_0$. Because we only work in the setting of cohomological field theories and not keep track of the calibration of them, we will just forget about these additional quadratic terms, introduced by the $S$--action.
    \end{remark}

\section{Givental's action form of $\SLC$ action}\label{section: aaa}
For any non--zero complex number $\sigma$ consider the particular upper--triangular group element $R^\sigma(z) := \exp\left(r^\sigma(z)\right)$ fixed by
$$
  r^\sigma(z) = 
  \left(
  \begin{array}{c c c}
    0 & \dots & \sigma \\
    \vdots & 0 & \vdots \\
    0 & \dots & 0
  \end{array}
  \right) z.
$$
For any function $f(\bt)$ we denote by $\left(f(\bt)\right)_p$ the expansion of it at the point $\bt = p$.
This section is devoted to the proof of the following theorem.
\begin{theorem}\label{th:FrobeniusManifold}
  Fix some
  $A = \begin{pmatrix}
	a & b \\
	c & d
      \end{pmatrix}
    \in {\rm SL}(2, \mathbb C)
  $, $\tau \in \CC$ and a CohFT with the partition function $\Z = \exp(\sum_{g \ge 0} \hbar^{g-1} \F_g)$ and small phase space potentials $F_g(\bt)$.
  \begin{itemize}
    \item[(a)] 
    Let $F_0(\bt)$ and $F_0^A(\bt)$ be convergent in some small neighborhoods of $p_1 := (0,\dots,0, A \cdot \tau)$ and of $p_2 := (0,\dots,0,\tau)$ respectively. Denote by $M$ and $M^A$ the Frobenius manifolds defined by $F_0$ and $F_0^A$ respectively. 
    Then we have an isomorphism:
    $$
      M^A \mid_{\bt = p_2} \ \cong \ \hat R^\sigma \cdot M \mid_{\bt = p_1},
    $$
    where $\sigma := -c(c\tau+d)$.
    
    \item[(b)] Fix some $g \ge 1$. Let $F_g(\bt)$ be convergent in a small neighborhood of $p_1$. Define the function $F_g^A$:
    \begin{equation}\label{equation: FgA definition}
      \begin{aligned}
      F_g^A(\bt) := (ct^n+d)^{2-2g} & F_g \left(\frac{t^1}{ct^n+d}, \dots, \frac{t^{n-1}}{ct^n+d}, \frac{at^n + b}{ct^n+d} \right) 
      \\
      & 
      + \frac{\delta_{1,g}}{2} \log \left( \frac{ct^n + d}{c\tau +d} \right).
      \end{aligned}
    \end{equation}
    If $F_g^A$ is convergent in a small neighborhood of $p_2$, then we have:
     $$
       [\hbar]^{g-1} \log \left( \hat R^\sigma \cdot \mathcal{Z}_{p_1} \right)(\tilde \bt) = (c\tau+d)^{2-2g} \ \left( F_g^A (\bt) \right)_{p_2}
     $$   
     where $\tilde \bt = \tilde \bt(\bt)$ is given by $\tilde t^{k,\alpha} = 0$ for all $k \ge 1$ and also: 
    $$
      \tilde t^{0,1} = t^1, \quad \tilde t^{0,\alpha} = (c\tau+d) \ t^k \quad \text{for} \quad 1 < k < n, \quad \tilde t^{0,n} = (c\tau+d)^2 \ t^n.
    $$
  \end{itemize}
\end{theorem}

Part~(a) of this theorem should be considered as an extension of the result of \cite{DbSS}, where it was proved for one particular $A$, giving the Inversion transformation of Dubrovin.

Part~(b) of this theorem only assumes one fixed $g$. Namely, one doesn't need the convergence of $g=1$ potential to make statement about the $g=2$ potential. This is a surprising fact, however it can also be very well seen if one applies $\hat R^\sigma$ by hands --- no genus $g' < g$ correlators of $\Z$ contribute to the genus $g$ correlators of $\hat R^\sigma \cdot \Z$ except when $g=1$.

The parameter $\tau$ of the theorem above plays an important role due to the following reasons. First of all for an arbitrary CohFT we can't assume any clear domain of holomorphicity (while in some particular examples it appears to be indeed large --- c.f. Section~\ref{subsection: GW}). Second reason is that the action of Givental deals with the correlators --- hence assumes the partition function to be represented by a power series.

To prove the theorem we first compute the $R$--action of $R^\sigma := \exp\left(r^\sigma(z) \right)$ explicitly considering $\sigma$ as a free parameter.

\begin{proposition}\label{prop: sigma action computation}
  Let $F_g(\bt)$ be a small phase space potentials of a unital CohFT on $(V,\eta)$ with $\eta_{ij} := \delta_{i+j,n+1}$.
  Let $F_g^\sigma(\bt)$ be a genus $g$ small phase potential of $\hat R^\sigma \cdot \Z(\hbar, \bt)$. We have:
  \begin{align*}
    F_g^\sigma(\bt) \ = \ &(1-\sigma t^n)^{2-2g}F_g \left(\frac{t^1}{1- \sigma t^n},\dots, \frac{t^n}{1- \sigma t^n} \right) 
    \\
    &+ \frac{\delta_{g,1}}{2} \log (1 - \sigma t^n)
    - \delta_{g,0} \frac{\sigma}{8(1 - \sigma t^n)} \left( \sum_{k=1}^n t^k t^{n+1-k} \right)^2. 
  \end{align*}
\end{proposition}
\begin{myproof}
  Consider the partition function $\Z^\sigma := \hat R^\sigma \cdot \Z$. Considering $\sigma$ as a parameter we compute the derivative of $\Z^\sigma$ w.r.t. it. Note that $r^\sigma(z) = \sigma \ r^1(z)$.
  By the definition of $\hat R^\sigma$ we have
  $$
    \frac{\p \Z^\sigma}{\p \sigma} = \frac{\p}{\p \sigma} \left( \exp(\sigma \ \hat r^1) \cdot \Z\right) = \hat r^1 \cdot \Z^\sigma,
  $$
  and we get a PDE on $\Z^\sigma$ with the initial condition $\Z^{\sigma = 0} = \Z$. Considering the formal series expansion in $\hbar$ on the both sides we get a series of PDEs on $F_g^\sigma$. By the definition we have:
  $$
    \hat r^1 = -\frac{1}{2} \hbar \frac{\p^2}{\p t^{0,1}\p t^{0,1}} + \sum_{l \ge 0} t^{l,n} \frac{\p}{\p t^{l+1,1}}.
  $$
  Consider the two summands of the differential operator written. Because the CohFT is unital, we have $\partial F_g/\partial t^{0,1} = 0$ for $g > 0$ (by U2 axiom and degree observations). Because of this, on the small phase space the second order summand gives non--zero contribution only applied to $F_0(\bt)$. We get:
  \begin{align*}
    \p_\sigma & \F_g^\sigma\mid_{t^{d,\alpha} = 0, d \ge 1} = \left( \sum_{l \ge 0} t^{l,n} \frac{\p \F_g^\sigma}{\p t^{l+1,1}} \right)\mid_{t^{d,\alpha} = 0, d \ge 1} - \frac{\delta_{g,0}}{2} \left(\frac{\p F_0^\sigma}{\p t^{0,1}} \right)^2
    - \frac{\delta_{g,1}}{2} \frac{\p^2 F_0^\sigma}{\p t^{0,1}\p t^{0,1}}
     \\
     &= t^{0,n} \left( (2g-2) F_g^\sigma + \sum_{k=1}^n t^{0,k} \frac{\p F_g^\sigma}{\p t^{0,k}} - \frac{\delta_{g,1}}{2} \right) - \ \delta_{g,0} \frac{1}{8} \left(\sum_{k=1}^n t^{0,k}t^{0,n+1-k} \right)^2,
  \end{align*}
  where we apply Dilaton equation and use the explicit form of the pairing $\eta$ on $V$.
  In what follows to simplify the notation we use $t^k := t^{0,k}$ for all $1 \le k \le n$.
  For $g \ge 2$ the PDE above is equivalent to:
  \begin{align*}
    \p_\sigma \log F_g^\sigma &= t^n \left( (2g-2) + \sum_{k=1}^n t^k \frac{\p}{\p t^k}\log F_g^\sigma \right),
    \\
    &=  t^n \sum_{k=1}^n t^k \frac{\p}{\p t^k} \left( \log F_g^\sigma + (2g-2) \log t^n\right).
  \end{align*}
  This PDE can be solved using the method of characteristics. We get:
  $$
    (t^n)^{2g-2} F_g^\sigma \left(t^1,\dots,t^n\right) = \left(\frac{t^n}{1 - \sigma t^n}\right)^{2g-2} F_g \left(\frac{t^1}{1 - \sigma t^n}, \dots, \frac{t^n}{1 - \sigma t^n}\right).
  $$
  For $g=1$ the PDE above reads:
  $$
    \p_\sigma F_1^\sigma = t^n \sum_{k=1}^n t^k \frac{\p}{\p t^k} \left( F_1^\sigma - \frac{1}{2}\log t^n \right).
  $$
  Using the method of characteristics again we get:
  $$
    F_1^\sigma(t^1,\dots,t^n) - \frac{1}{2}\log t^n = F_1\left(\frac{t^1}{1 - \sigma t^n}, \dots, \frac{t^n}{1 - \sigma t^n} \right) - \frac{1}{2}\log \left(\frac{t^n}{1-\sigma t^n}\right).
  $$
  For $g = 0$ it's easy to see by the direct differentiation that the function $F_0^\sigma$ satisfies the PDE above with the initial condition $F_0^{\sigma = 0} = F_0$.
  This completes the proof.
\end{myproof}
\begin{remark}
  The first version of this proposition was proved via the graph counting technique introduced by \cite{DbSS}. 
  The author is grateful to the anonymous referee for proposing this short proof.
\end{remark}

\begin{myproof}[of Theorem~\ref{th:FrobeniusManifold}]
  First note that the function $F_0^\sigma$ of Proposition~\ref{prop: sigma action computation} satisfies:
  $$
    F_0^\sigma = F_0^A \quad \text{for} \quad A = \begin{pmatrix} 1 & 0 \\ -\sigma & 1\end{pmatrix}.
  $$ 
  Consider the changes of the variables $\hat t^n = (c\tau +d)^2 t^n$ and $\tilde t^n = \hat t^n + \tau$. We have:
  $$
    \frac{t^n}{1 +c(c\tau+d)t^n} + \frac{a\tau + b}{c\tau +d} = \frac{a(\hat t^n + \tau) + b}{c(\hat t^n + \tau) + d} = \frac{a\tilde t^n + b}{c\tilde t^n +d}.
  $$
  By using Proposition~\ref{prop: sigma action computation} we get the isomorphism of the Frobenius manifolds of the part (a). Part (b) follows immediately from Proposition~\ref{prop: sigma action computation} by applying the changes of the variables given.
\end{myproof}

\section{Coordinate--free form of the $\SLC$ action}\label{section: coordinateFree}

  Part~(b) of Theorem~\ref{th:FrobeniusManifold} makes use of the functions expanded at a certain points. It's natural to get rid of this special requirement.
  Let $S^c = \exp (s_1^c)$ and $S_0^A$ be Givental's $S$--actions for: 
  $$
    s_1^c := 
      \begin{pmatrix}
	0 & \dots & 0 \\
	\vdots & 0 & \vdots \\
	c & \dots & 0
      \end{pmatrix} z^{-1},
    \quad
    S_0^A := 
      \begin{pmatrix}
	1 & \dots & 0 \\
	\vdots & (c\tau+d) I_{n-2} & \vdots \\
	0 & \dots & (c\tau+d)^2
      \end{pmatrix} \quad \text{for} \quad
      A = \begin{pmatrix}
         a & b \\ c & d
        \end{pmatrix}.
  $$
  where $S_0^A$ acts by $(\hbar,\bt) \to ((c\tau+d)\hbar, \tilde \bt)$ with $\tilde t^{\alpha} = \left(S_0^A\right)_\beta^\alpha t^{\beta}$.
  On the small phase space $\hat S^c$ acts as a shift of the coordinate $t^n$ (recall that we restrict ourselves to the space of cohomological field theories and only assume the $S$--action up to quadratic terms --- see Remark~\ref{remark: calibration}). 
  Hence the statement of Theorem~\ref{th:FrobeniusManifold} rewrites as:
  \begin{equation*}
    \hat S_0^A \cdot \hat S^{\tau} \cdot F^A_g =  \hat R^\sigma  \cdot \hat S^{A \cdot \tau} \cdot F_g.
  \end{equation*}
  This equation suggests the following Givental analog of the $\SLC$--action:
  \begin{align*}
    & \hat A_G^\tau := \hat S^{-\tau} \left( \hat S_0^A \right)^{-1} \cdot  \hat R^\sigma  \cdot \hat S^{A \cdot \tau}.
  \end{align*}  
  It's easy to check by hands that $\SLC$ action defined in the analytic form of Eq.\eqref{eq:SLAction} and Eq.\eqref{equation: FgA definition} is indeed a group action. This fact is not clear on the Givental's side --- the upper--triangular and lower--triangular groups form two different groups, whose elements do not commute, and $\hat A_G^\tau$ makes use of both of them. 
  
  Compared again to the analytic action of Eq.\eqref{eq:SLAction} and Eq.\eqref{equation: FgA definition}, the action $\hat A_G^\tau$ is preferable if one works with the cohomological field theories and not just Frobenius manifolds, in particular in mirror symmetry. Let's illustrate this on the example.   
  
  \begin{example}
    For a given CohFT $\Lambda_{g,n}$ with a partition function $\Z$ let $\Z = \hat R' \cdot \Z'$ for some other CohFT partition function $\Z'$ and an upper--triangular group element $R'$. For a fixed $A \in \SLC$ let $\tau'$ be s.t. $A \cdot \tau' = 0$. Then we have:
    $$
      \hat A_G^{\tau'} \cdot \Z = \hat S^{-\tau'} \left( \hat S_0^A \right)^{-1} \cdot  \left( R^\sigma R' \right)\widehat{\ } \cdot \Z',
    $$
    where we used the fact that the upper--triangular group elements form indeed a group.
  \end{example}

  A big issue in the definition above is the composition $\hat A_G^{\tau_0} \cdot \hat B_G^{\tau_1}$ for some $A,B \in \SLC$ and $\tau_0,\tau_1 \in \CC$.
  In what follows we show that the set of all differential operators $\hat A_G^\tau$ forms a groupoid. Recall the definition of it.

  \begin{definition}
    The set $H$ is called groupoid if it's equipped with the unary operation ${}^{-1}: H \to H$ and a partial binary operation $\ast:  H\times H \rightharpoonup H$, s.t. the following conditions are satisfied for all $a,b,c \in H$.
    \begin{description}
     \item[(1)] $a \ast a^{-1}$ and $a^{-1} \ast a$ are defined for all $a \in H$,
     \item[(2)] $a \ast b$ is defined if and only if $a^{-1} \ast a = b\ast b^{-1}$,
     \item[(3)] if $a\ast b$ and $b\ast c$ are defined, then $(a\ast b)\ast c$ and $a\ast (b\ast c)$ are defined and equal,
     \item[(4)] each of $a^{-1}\ast a\ast b$, $b\ast a^{-1}\ast a$, $a\ast a^{-1}\ast b$, and $b\ast a\ast a^{-1}$ is equal to $b$ if it is defined. 
    \end{description}

  \end{definition}

    In what follows denote also by $\hat{\mathrm{Id}}_G$ the identity operator.

  \begin{proposition}
    The set $\SLC\times \CC = \{\hat A_G^\tau \ | \ A \in \SLC, \ \tau \in \CC\}$ is a groupoid with 
    \begin{description}
      \item[(i)] $\left( \hat A_G^\tau \right)^{-1} := \widehat{\left( A^{-1} \right)}^{A \cdot \tau}_G$ for all $A \in \SLC$ and $\tau \in \CC$,
      \item[(ii)] $\hat A_G^{\tau_0}\ast \hat B_G^{\tau_1} := \hat A_G^{\tau_0}\cdot \hat B_G^{\tau_1} = \widehat{\left(BA \right)}^{\tau_0}_G$ for all $A,B \in \SLC$ and $\tau_0,\tau_1$, s.t. $A\cdot \tau_0 = \tau_1$.
    \end{description}
  \end{proposition}
  \begin{myproof}
    We show the conditions (1)--(4) of the definition above. 
    
    Condition (1) is straightforward because for any $\hat A_G^{\tau}$ we have $A^{-1} \cdot ( A \cdot \tau) = \tau$. In condition (2) we should consider $\hat A_G^{\tau_0} \ast \hat B_G^{\tau_1}$. To show this condition consider the following lemma.

    \begin{lemma}
      For any $A \in \SLC$ the following formula holds:
      $$
	\hat A_G^{\tau_0} \cdot \widehat{\left( A^{-1}\right)}^{A\cdot \tau_0}_G  = \widehat{\mathrm{Id}}_G^{\tau_0}.
      $$
    \end{lemma}
    \begin{myproof}
      Let $A = \begin{pmatrix}a & b \\ c & d\end{pmatrix} \in \SLC$.
      Denote $\sigma_A := - c(c\tau_0 + d)$ and $\sigma^\prime_A := - c/(c\tau_0+d)$.
      Note that we have then $( \hat S_0^A )^{-1} \cdot \hat R^{\sigma_A} = \hat R^{\sigma_A^\prime}  \cdot ( \hat S_0^A )^{-1}$.
      For $\tau_1 = A\cdot \tau_0 $. Then we have:
      $$
	\sigma_A = -\sigma^\prime_{A^{-1}}, \ \sigma^\prime_A = -\sigma_{A^{-1}}, \quad S_0^A = \left( S_0^{A^{-1}} \right)^{-1}.
      $$
      The composition of Givental's actions reads:
      \begin{align*}
	\hat A_G \cdot \left( A^{-1}\right)\widehat{\ }_G & = \hat S^{-\tau_0} \left( \hat S_0^A \right)^{-1} \cdot  \hat R^{\sigma_A}  \cdot \hat S^{A \cdot \tau_0} \cdot \hat S^{-\tau_1} \left( \hat S_0^{A^{-1}} \right)^{-1} \cdot  \hat R^{\sigma_{A^{-1}}}  \cdot \hat S^{A^{-1} \cdot \tau_1}
	\\
	& = \hat S^{-\tau_0} \left( \hat S_0^A \right)^{-1} \cdot  \hat R^{\sigma_A}  \cdot \hat R^{\sigma^\prime_{A^{-1}}}  \cdot   \left( \hat S_0^{A^{-1}} \right)^{-1} \cdot \hat S^{\tau_0} = \mathrm{Id}.
      \end{align*}
    \end{myproof}
    
    Hence the equality $\left( \hat A_G^{\tau_0} \right)^{-1} \ast \hat A_G^{\tau_0} = \hat B_G^{\tau_1} \ast \left( \hat B_G^{\tau_1} \right)^{-1}$ is equivalent to $\widehat{\mathrm{Id}}_G^{A\cdot\tau_0} = \widehat{\mathrm{Id}}_G^{\tau_1} $ that holds if and only if $\tau_1 = A \cdot \tau_0$. This concludes Condition (2).
    
    To show Condition (3) we start with the following lemma.

    \begin{lemma}
      Let $A,B \in \SLC$ and $\tau_1 = A\cdot \tau_0$. Then $\hat A_G^{\tau_0} \cdot \hat B_G^{\tau_1} = \hat C_G^{\tau_0}$ for $C = BA$.
    \end{lemma}
    \begin{myproof}
      Let $A = \begin{pmatrix}
		a_{11} & a_{21} \\ a_{21} & a_{22}
	      \end{pmatrix}$,
      $B = \begin{pmatrix}
		b_{11} & b_{21} \\ b_{21} & b_{22}
	      \end{pmatrix}$. Then canceling out two $S$--actions we can write:
      $$
	\hat A_G \cdot \hat B_G = \hat S^{- \tau_0} \left( \hat S_0^A \right)^{-1} \hat R^{\sigma_A} \left(\hat S_0^B \right)^{-1} \hat R^{\sigma_B} \hat S^{B\tau_1}.
      $$
      Commuting $\left(\hat S_0^B \right)^{-1}$ with $\hat R^{\sigma_A}$ as in the proof of the lemma above and observing that $\hat S_0^B \hat S_0^A = \hat S_0^{C}$ we get:
      $$
	\hat A_G \cdot \hat B_G = \hat S^{- \tau_0} \left( \hat S_0^C \right)^{-1} \hat R^{\tilde \sigma + \sigma_B} \hat S^{C\tau_0},
      $$
      for $\tilde \sigma = \sigma_A (b_{21} \tau_1 + b_{22})^2$. Some easy but long computations give the needed equality $\tilde \sigma +  \sigma_B = \sigma_C$.
    \end{myproof}

    For any $A,B,C \in \SLC$ and $\tau_0,\tau_1,\tau_2 \in \CC$ the products $\hat A_G^{\tau_0} \ast \hat B_G^{\tau_1}$ and $\hat B_G^{\tau_1} \ast C_G^{\tau_2}$ are defined if and only if $\tau_1 = A\cdot \tau_0$ and $\tau_2 = B\cdot \tau_1$. Then $(\hat A_G^{\tau_0} \ast \hat B_G^{\tau_1}) \ast \hat C_G^{\tau_2} = \widehat{BA}_g^{\tau_0} \ast \hat C_G^{\tau_2}$ is defined if and only if $\tau_2 = BA\cdot \tau_0 = B \cdot (A \cdot \tau_0) = B \cdot \tau_1$. This concludes Condition (3).
    
    Proof of Condition (4) is straightforward by using the two lemmas above.
    \end{myproof}

    In order to consider the action of $\hat A_G^\tau$ on the space of partition functions of the cohomological field theories one has to resolve the following problem. The $S$--action applies the shift of the variables to an infinite series. This operation can produce divergent power series and one can't apply the consequent $R$--action then. Moreover it could happen that for some $A,B \in \SLC$, $\tau_0,\tau_1 \in \CC$ and a CohFT partition function $\Z$ the action $\widehat{BA}_G^{\tau_0} \cdot \Z$ is defined, but $\hat B_G^{\tau_1} \cdot \left( \hat A_G^{\tau_0} \cdot \Z \right)$ --- not (exactly due to the convergence issues). Hence one could only consider the action of some subgroupoid of all $\hat A_G^{\tau_0}$ on the certain subset of the full Fock space of partition functions, satisfying some convergence properties. This subject is beyond the purposes of this work and we postpone it to the future.
    
\section{$\SLC$ action in singularity theory}\label{section: totalAncestorPotential}

  This section is devoted to the connection of the $\SLC$ action developed above to the action on the total ancestor potential of a simple--elliptic singularity, that changes a primitive form.   
    \subsection{Frobenius manifold Saito--Givental theory of a hypersurface singularity} Let $W(\bx) \in \mathcal{O}_{\CC^N,0}$ define an isolated singularity at $0 \in \mathbb C^N$ and $\phi_i(\bx) \in \mathcal O_{\mathbb C^N,0}$ form a basis, generating the Milnor algebra of the singularity: $\langle \phi_1, \dots, \phi_\mu \rangle = \mathcal O_{\mathbb C^N, 0} / \left\langle \partial_{x_1} W, \dots, \partial_{x_N}W \right\rangle$. The \textit{unfolding} of the singularity is the following function:
    \[
      W(\bx, \bs) := W(\bx) + \sum_{i=1}^\mu \phi_i(\bx) s^i,
    \]
    where $\mu$ is the Milnor number of $W(\bx)$ and the coordinates $s^1,\dots,s^\mu$ belong to the \textit{base space} $\mathcal S$ of the unfolding. 
    The sheaf $\T_{\mathcal S}$ can be endowed with an algebra structure. 
    Fixing a volume form $\omega = f(\bx, \bs) dx_1 \dots dx_N$ one can introduce a non--degenerate bilinear form $\eta$ on $\T_{\mathcal S}$ --- the residue pairing.     
    Special choice of the volume form is needed to make $\eta$ flat and get the Frobenius manifold.
    
    \begin{theorem}[\cite{S}]
      There is a choice of the volume form $\omega$ such that the pairing $\eta$ is flat and defines together with $c_{ij}^k(\bs)$ a structure of a Frobenius manifold on $\mathcal S$.
    \end{theorem}
    
    Such volume form is called a \textit{primitive form} of Saito and its existence is another complicated question.
    Its choice is generally not unique. 

    By considering the asymptotic expansions of a certain oscillatory integrals near some $\bs \in \mathcal{S}$ one can construct an upper--triangular group element $R_\bs$ and also a partition function $\mathcal{A}_{\bs,\zeta} (\hbar,\bt)$, called a \textit{total ancestor potential of the singularity $W(\bx)$}. This data depends also on $\zeta$ and the choice of $\bs \in \mathcal{S}$. Comparing to the case of the CohFTs discussed in Section~\ref{section: GiventalsAction} here we get a family of partition functions, depending on this data.     
    
    The function
    $$
      F_{\bs,\zeta}(\bt) := \left[ \hbar^{-1} \right] \ \log \mathcal{A}_{\bs,\zeta} (\hbar, \bt) |_{\bt^{d,p} = 0, \ d \ge 1}
    $$
    is then the Frobenius manifold potential of the theorem above. It defines a Frobenius manifold structure on $\mathcal S$ in a neighborhood of the point $\bs$ with the primitive form $\zeta$. We will denote it by $M_{\zeta,\bs}$.

    \subsection{Choice of the primitive form}
    Let $W(\bx) = W_\sigma(\bx)$ be one of the following polynomials, depending also on a complex parameter $\sigma$:
    \begin{align*}
      & x_1^3 + x_2^3 + x_3^3 + \sigma x_1x_2x_3,
      \\
      & x_1^2x_3 + x_1x_2^3 + x_3^2 + \sigma x_1x_2x_3,
      \\
      & x_1^3x_3 + x_2^3 + x_3^2 + \sigma x_1x_2x_3.
    \end{align*}
    These polynomials define the so--called simple--elliptic singularities. Their zero sets in certain weighted projective spaces define the families of the elliptic curves $E_\sigma$.   
    The primitive form $\zeta$ for the simple--elliptic singularities was found by \cite[Paragraph 3]{S}: 
    $$
      \zeta = \frac{d^3\bx}{\pi_A(\sigma)},
    $$
    where $\pi_A(\sigma)$ is a solution to the Picard-Fuchs equation of $E_\sigma$. Different solutions $\pi_A(\sigma)$ give different primitive forms and (generally) different Frobenius manifold structures. 

    This property was explored differently in \cite{MR} and \cite{BT} to consider different primitive forms of the fixed simple elliptic singularity. In particular, the second approach was to use the same $\SLC$--action as we described above, but with a particular matrix $A^{(\tau_0,\omega_0)}$ for all fixed $\tau_0 \in \HH$ and $\omega_0 \in \CC^*$:
    $$
      A^{(\tau_0, \omega_0)} := 
      \begin{pmatrix}
	  \dfrac{\sqrt{-1}\bar{\tau}_0}{2\omega_0{\rm Im}(\tau_0)} & \omega_0 \tau_0
	  \\
	  \dfrac{\sqrt{-1}}{2\omega_0{\rm Im}(\tau_0)} & \omega_0
	\end{pmatrix}.
    $$
    The $\SLC$--action of such matrix on the space of Frobenius manifolds as in Example~1 of Introduction can be considered as a model for the primitive form change (see Section~2.5, \cite{BT}).
    \begin{remark}
      It was a $\mathrm{GL}(2,\CC)$ matrix, used there, however with a fixed determinant $1/(2 \pi \sqrt{-1})$. We rescale the matrix here in order to get a $\SLC$--action. This rescaling is equivalent to the variable rescaling.
    \end{remark}

  We are going to compare it with the completely different approach of Milanov--Ruan.
    
    \subsubsection{Approach of Milanov--Ruan}
    From now on fix $n := \mu$, the Milnor number of the singularity.
    Let $e_1, \dots, e_n$ be the basis vectors of $T_{\bs}\mathcal{S}$ at ${\bs} = (s^1,\dots,s^n)$.
    For any 
    $A \in \SLC$ let the linear operator $J: T_\bs \mathcal{S} \to T_\bs \mathcal{S}$ be defined by:
    $$
      J (e_1) = e_1, \ J(e_n) = (cs^n + d)^2 e_n, \ J(e_p) = (cs^n + d) e_p, \quad 1 < p < n, \quad A = \begin{pmatrix}
	  a & b
	  \\
	  c & d
	\end{pmatrix}.
    $$
    Let $J$ act on $\bt$ by:
    $t^{d,j} \rightarrow \sum_{i=1}^n J_i^j t^{d,i}$ for all $d \ge 0$.
    Consider Givental's upper--triangular group element $X_{\textbf s}(z)$ defined by:
    $$
      \left( X_\bs(z) \right)_i^j = \delta_i^j - z \frac{c}{c s^n + d} \delta_{i,n}\delta^{j,1},
    $$
  
    \begin{theorem}[Theorem~4.4 in \cite{MR}]\label{theorem: Milanov-Ruan}
      Consider two total ancestor potentials of the same simple--elliptic singularity~\footnote{Only one particular case of simple elliptic singularity $P_8$ was considered in \cite{MR}. However the technique used is extended in a straightforward way to all other cases too.} 
      with the primitive forms $\zeta_1$ and $\zeta_2$. Then there is $A \in {\rm SL}(2,\mathbb C)$ such that the corresponding total ancestor potentials are connected by the following transformation\footnote{The statement of theorem of Milanov--Ruan has the operator $\hat X_\bs$ rather than the inverse of it, however they also consider a bit different definition of the Givental's action --- the connection to our is by taking the inverse $R$--action.}:
      $$
	\mathcal A_{\zeta_1, A \cdot \bs}( \hbar, \bt) = \left( \left(\hat X_\bs\right)^{-1} \cdot \mathcal A_{\zeta_2, \bs}\right) \left( (c \tau + d)^2 \hbar, J \bt \right),
      $$
      where 
      $\bs = (0, \dots, 0, \tau)$ and $A \cdot \bs = (0, \dots, 0, A \cdot \tau)$.
    \end{theorem}
    In the theorem above it appears to be very hard to present particular $\SLC$ matrix $A$.

    \subsection{Equivalence of the approaches.}
    Comparing the total ancestor potential formula of Milanov--Ruan with the $\SLC$ action developed in this paper we get the theorem.
    \begin{theorem}\label{theorem: MR-BT equivalence}
      Let the two total ancestor potentials $\mathcal{A}_{\zeta_1,\bs}$, $\mathcal{A}_{\zeta_2,\bs^\prime}$ of the same simple--elliptic singularity $W(\bx)$ be connected by a symplectic transformation of Milanov--Ruan. Then the potentials of the Frobenius manifolds $M_{\zeta_1,\bs}$, $M_{\zeta_2,\bs'}$ corresponding to these total ancestor potentials are connected by the formula \eqref{eq:SLAction} for some $A \in \SLC$.
    \end{theorem}
    \begin{myproof}
      Consider particular pair of primitive forms $\zeta_1$ and $\zeta_2$ with the corresponding $\SLC$--matrix $A$ as in Theorem~\ref{theorem: Milanov-Ruan}. We show that the transformation of Milanov-Ruan acts in the same way as Givental's form of our $\SLC$ action.

      Denote by $\hat J$ the rescaling action $(\hbar, \bt) \to ((c\tau+d)^2 \hbar, J \bt)$. Note that its action is the same as $S_0^A$ action of Section~\ref{section: coordinateFree}. Recall $R^\sigma$--action of Section~\ref{section: aaa}. Because of its particular form we have:
      $$
	R^\sigma(z) = \exp(r^\sigma z) = 1 + r^\sigma z.
      $$
      We know that for $\sigma := -c(c \tau + d)$ the following equality holds:
      $$
	\hat R^\sigma \cdot \hat S_0^A = \hat J \cdot \hat X_\bs.
      $$
      Hence the transformation of Milanov--Ruan is represented by the $\SLC$ action we consider. We only have to compare the points of the Frobenius manifolds on both sides. Using Theorem~\ref{theorem: Milanov-Ruan} and Theorem~\ref{th:FrobeniusManifold} we write:
      $$
	F_{\zeta_1, A \cdot \tau} = \left( \hat R^\sigma \right)^{-1} \hat S_0^A F_{\zeta_2, \tau} =  \hat R^{-\sigma} \hat S_0^A F_{\zeta_2, \tau}
	=  \hat S_0^A \left( \hat X_\bs \right)^{-1} F_{\zeta_2, \tau}.
      $$
      This completes proof of the theorem.
    \end{myproof}
    \begin{corollary}
      The approaches of \cite{MR} and \cite{BT} for the action changing the primitive form of a simple--elliptic singularity coincide.
    \end{corollary}
    This result shows that taking all the possible primitive forms for a fixed simple--elliptic singularity there are only two ``geometrically'' different Frobenius manifolds --- $F_{\zeta_1}$ and $\hat I \cdot F_{\zeta_1}$. All the other Frobenius manifolds fixed by other choices of the primitive forms are obtained from these two by taking linear changes of the variables.



\begin{thebibliography}{9}

\bibitem{B_hf}
Basalaev, A.
\newblock {Orbifold GW theory as the Hurwitz--Frobenius submanifold}.
\newblock {\em Journal of Geometry and Physics}, 77:30--42 (2014).

\bibitem{B_o}
Basalaev, A.
\newblock {$\SLC$--action on cohomological field theories and Gromov--Witten theory of elliptic orbifolds}.
\newblock {\em Oberwolfach reports}, 22/2015 (2015).

\bibitem{B6dim}
Basalaev, A.
\newblock {6--dimensional FJRW theories of the simple--elliptic singularities}.
\newblock {\em arXiv preprint: 1610.07428 (2016)}.

\bibitem{B_hp}
Basalaev, A.
\newblock {Homepage}.
\newblock \href{http://basalaev.wordpress.com}{http://basalaev.wordpress.com}.

\bibitem{BP}
Basalaev, A., Priddis, N.
\newblock {Givental-type reconstruction at a non-semisimple point }.
\newblock {\em arXiv preprint: 1605.07862v2, accepted by Mich. Math. J. (2017)}.

\bibitem{BT}
Basalaev, A., Takahashi, A.
\newblock {On rational Frobenius Manifolds of rank three with symmetries}.
\newblock {\em J. Geom. Phys.}, (77):30--42 (2014).
 
\bibitem{CIR}
Chiodo, A., Iritani, H., Ruan, Y.
\newblock {Landau-Ginzburg/Calabi-Yau correspondence, global mirror symmetry and Orlov equivalence}.
\newblock {\em Publ.math.IHES}, (2014):119--127 (2014).

\bibitem{D}
Dubrovin, B.
\newblock {Geometry of 2d topological field theories}.
\newblock In {\em Lecture Notes in Math}, pages 120--348. Springer (1996).

\bibitem{DZ}
Dubrovin, B., Zhang, Y.
\newblock {Bihamiltonian Hierarchies in 2D Topological Field Theory At One-Loop
  Approximation}.
\newblock {\em arXiv preprint: hep-th/9712232 (2001)}.

\bibitem{DbSS}
Dunin--Barkowski, P., Shadrin, S., and Spitz, L.
\newblock {Givental Graphs and Inversion Symmetry}.
\newblock {\em Letters in Mathematical Physics}, 103(5):533--557 (2013).

\bibitem{FSZ}
Faber, C., Shadrin, S., and Zvonkine, D.
\newblock {Tautological relations and the $r$-spin Witten conjecture}.
\newblock {\em Ann. Sci. Ec. Norm. Super.}, 4(43):621--658 (2006).

\bibitem{G}
Givental, A.
\newblock {Gromov--Witten invariants and quantization of quadratic hamiltonians}.
\newblock {\em Mosc. Math. J.}, 1(4):1--17 (2001).

\bibitem{H}
Hertling, C.
\newblock {\em {Frobenius Manifolds and Moduli Spaces for Singularities}}.
\newblock Cambridge University Press, Cambridge, cambridge edition (2002).

\bibitem{KM}
Kontsevich, M., Manin, Y.
\newblock {Gromov--Witten classes, quantum cohomology, and enumerative   geometry}.
\newblock {\em Comm. Math. Phys.}, 164(3):525--562 (1994).

\bibitem{L}
Lee, Y.P.
\newblock {Notes on axiomatic Gromov--Witten theory and applications}.
\newblock {\em Proceedings Of Symposia In Pure Mathematics}, 80 (2005).

\bibitem{M}
Milanov, T.
\newblock {Analyticity of the total ancestor potential in singularity theory}.
\newblock {\em Adv. in Math.}, 255: 217--241 (2014).

\bibitem{MR}
Milanov, T., Ruan, Y.
\newblock {Gromov--Witten theory of elliptic orbifold $\mathbb{P}^{1}$ and quasi--modular forms}.
\newblock {\em arXiv preprint:1106.2321} (2011).

\bibitem{MS2}
Milanov, T., Shen, Y.
\newblock {The modular group for the total ancestor potential of Fermat simple elliptic singularities}.
\newblock {\em Commun. Number Theory Phys}, 329--368 (2014).

\bibitem{PPZ}
Pandharipande, R., Pixton, A., Zvonkine, D.
\newblock {Relations on $\M{g}{n}$ via 3--spin structures}.
\newblock {\em Amer. Math. Soc.}, (28):279--309 (2015).

\bibitem{S} 
Saito, K.
\newblock {Period mapping associated to a primitive form}.
\newblock {\em Publ RIMS, Kyoto Univ.}, (19):1231--1264 (1983).

\bibitem{ST}
Satake, I., Takahashi, A.
\newblock {Gromov--Witten invariants for mirror orbifolds of simple elliptic singularities}.
\newblock {\em Ann. Inst. Fourier}, 61:2885--2907 (2011).

\bibitem{Sh}
Shadrin, S.
\newblock {BCOV theory via Givental group action on cohomological field theories}.
\newblock {\em Mosc. Math. J.}, 9:411--429 (2009).

\bibitem{SZ}
Shen, Y., Zhou, J.
\newblock {Ramanujan Identities and Quasi-Modularity in Gromov-Witten Theory}.
\newblock {\em arXiv preprint: 1411.2078v1} (2014).

\bibitem{T}
Teleman, C.
\newblock {The structure of 2D semi-simple field theories}.
\newblock {\em Invent. Math.}, (189):525--588 (2012).

\bibitem{Z} 
D. Zagier, 
\newblock {Elliptic modular forms and their applications}.
\newblock In {\em 1-2-3 Modul. forms}, Springer Universitext, (2008), pp. 1--103.

\end{thebibliography}
\end{document}